\documentclass[letterpaper,11pt]{article}
\usepackage{amssymb,fullpage}
\usepackage{times,upgreek}

\usepackage{xspace}

\usepackage{graphicx}
\usepackage{amsmath}
\usepackage{amsfonts}
\usepackage{amssymb}

\newtheorem{theorem}{Theorem}

\newtheorem{corollary}[theorem]{Corollary}

\newtheorem{lemma}[theorem]{Lemma}

\newtheorem{proposition}[theorem]{Proposition}

\newenvironment{proof}[1][Proof]{\textbf{#1.} }{\ \rule{0.5em}{0.5em}}
\begin{document}

\newcommand{\del}{{\delta}}
\newcommand{\A}{{\cal{A}}}
\newcommand{\B}{{\cal{B}}}
\newcommand{\C}{{\cal{C}}}
\newcommand{\R}{{\cal{R}}}
\newcommand{\eps}{{\varepsilon}}

\date{}

\title{Colorful bin packing}

\author{Gy\"{o}rgy D\'{o}sa\thanks{Department of Mathematics, University of Pannonia,
Veszpr\'{e}m, Hungary, \texttt{dosagy@almos.vein.hu}.}
\and Leah Epstein\thanks{ Department of Mathematics, University of Haifa,
Haifa, Israel. \texttt{lea@math.haifa.ac.il}. }}

\maketitle

\begin{abstract}
We study a variant of online bin packing, called colorful bin packing. In this problem, items that are presented one by one are to be packed into bins of size $1$. Each item $i$ has a size $s_i \in [0,1]$ and a color $c_i \in \C$, where $\C$ is a set of colors (that is not necessarily known in advance). The total size of items packed into a bin cannot exceed its size, thus an item $i$ can always be packed into a new bin, but an item cannot be packed into a non-empty bin if the previous item packed into that bin has the same color, or if the occupied space in it is larger than $1-s_i$.
This problem generalizes standard online bin packing and online black and white bin packing (where $|\C|=2$). We prove that colorful bin packing is harder than black and white bin packing in the sense that an online algorithm for zero size items that packs the input into the smallest possible number of bins cannot exist for $|\C|\geq 3$,
while it is known that such an algorithm exists for $|\C|=2$.
We show that natural generalizations of classic algorithms for bin packing fail to work for the case $|\C|\geq 3$, and moreover, algorithms that perform well for black and white bin packing do not perform well either, already for the case $|\C|=3$. Our  main results are a new algorithm for colorful bin packing that we design and analyze, whose absolute competitive ratio is $4$, and a new lower bound of $2$ on the asymptotic competitive ratio of any algorithm, that is valid even for black and white bin packing.
\end{abstract}

\section{Introduction}
Colorful bin packing is a packing problem where a sequence of colored items is presented to the algorithm, and the goal is to partition (or pack) the items into a minimal number of bins. The set of items is denoted by $\{1,2,\ldots,n\}$, where
$0 \leq s_i \leq 1$ is the size of item $i$, and $c_i \in \C$ is its color. The items are to be packed one by one (according to their order in the input sequence), such that the items packed into each bin have a total size of at most $1$, and any two items packed consecutively into one bin have different colors. Since the input is viewed as a sequence rather than a set, the natural scenario for this problem is an online one; after an item has been packed, the next item is presented. In an online environment, the algorithm packs an item without any knowledge regarding the further items, and the set $\C$ (or even its cardinality) is not necessarily known to the algorithm. The number of items, $n$, is typically unknown to the algorithm as well. In the case that inputs are viewed as sequences and not as sets, online algorithms are typically compared to optimal offline algorithms that must pack the items exactly in the same order as they appear in the input.

Consider an input for colorful bin packing with $N$ red items of size zero, followed by $N$ blue items of size zero. This input requires $N$ bins, but reordering the items reduces the required number of bins to $1$. Thus, distinguishing reasonable online algorithms from less successful ones cannot be done by comparison to offline algorithms that are allowed to reorder the input. The offline algorithms to which we compare our online algorithm are therefore not allowed to reorder the input.
Such an optimal offline algorithm is denoted by OPT (OPT denotes a specific optimal offline algorithm, and we use OPT to denote also the number of bins that it uses for a given input).
The absolute competitive ratio of an algorithm is the supremum ratio over all inputs between the number of bins that it uses and the number of bins that OPT uses (for the same input). The asymptotic competitive ratio is the limit of absolute competitive ratios $R_K$ when $K$ tends to infinity and $R_K$ takes into account only inputs for which OPT uses at least $K$ bins. Note that (by definition), for a given algorithm (for some online bin packing problem), its  asymptotic competitive ratio never exceeds its absolute competitive ratio.

The special case of colorful bin packing, called black and white packing, was introduced in \cite{waoa12}. In this variant there are just two colors, called black and white. The motivation for black and white bin packing was in assignment to containers of items so that any two items packed consecutively into one bin can be easily distinguished later. An example for such items was articles that are printed on either white paper or recycled paper, in which case bins simply contain piles of paper, and packing articles printed on the two kinds of paper so that the two kinds alternate allows to distinguish them easily.
Colorful bin packing is the generalization where there is a number of different kinds of printing paper (for example, paper of distinct colors that is used for printing advertisement flyers), and in order to distinguish between two items (two piles of flyers), they have to have different colors of printing paper.

It was shown \cite{waoa12} that the natural generalizations of several well-known algorithms fail to obtain finite competitive ratios. For example, Next Fit (NF) for colorful bin packing (and for black and white bin packing) packs items into a single active bin, and moves to a new active bin as soon as packing an item into the active bin is impossible. For standard bin packing, a new active bin is opened when there is no space for the new item in the previous active bin, but for colorful bin packing a new bin will be opened either in this case, or when the last item of the active bin and the new item have the same color.  It was shown in \cite{waoa12} that this algorithm fails to achieve a finite competitive ratio (already for two colors). Harmonic algorithms \cite{LeeLee}, that partition items into sub-inputs according to sizes and pack each sub-input independently of the other sub-inputs, were also shown to have unbounded competitive ratios \cite{waoa12}.
On the other hand, there are some basic online bin packing algorithms that can be adapted successfully for black and white bin packing. The generalizations of Any Fit (AF) algorithms, that never use a new bin unless there is not other way to pack a new item, were shown to have constant absolute competitive ratios. The generalized versions of such algorithms for colorful bin packing open a new bin only if the current item cannot be packed into an existing bin such that the color constraint is kept and the total size of items packed into the bin will remain at most $1$. Three important special cases of AF are First Fit (FF), Best Fit (BF), and Worst Fit (WF). These algorithms select the bin where a new item is packed (out of the feasible options) to be the bin of minimum index, the a bin with the smallest empty space, and a bin with the largest empty space, respectively. The difference with classical bin packing is that the infeasible bins can be of two kinds, either those that do not have sufficient empty space, and those where the last packed item has the same color as the color of the new item.  It was shown that all AF algorithms have absolute and asymptotic competitive ratios of at least $3$ and at most $5$ for black and white bin packing. Vesel\'y \cite{Ves} tightened the bound and showed an upper bound of $3$ on the absolute competitive ratio of AF algorithms. The results of \cite{waoa12,Ves} in fact show that the absolute competitive ratio of WF is $2+\frac{1}{d-1}$, if all items have sizes in $(0,\frac 1d]$ (while FF and BF still have absolute and asymptotic competitive ratios of exactly $3$ even in this restricted case). The positive results for AF algorithms are valid only for black and white packing but not for colorful bin packing. In contrast to these last results, we will show that AF algorithms do not have constant (absolute or asymptotic) competitive ratios for colorful bin packing with $|\C|\geq 3$.

Colorful bin packing is also a generalization of standard bin packing (since already black and white bin packing is such a generalization). For standard bin packing, NF has an asymptotic and an absolute competitive ratio of $2$ \cite{J74}. Any Fit algorithms all have absolute competitive ratios of at most $2$ \cite{U71,John,J74,johnson1974worst,DS12} (some of these algorithms have smaller absolute or asymptotic competitive ratios; for example, in \cite{DS12} it is shown that FF has an absolute competitive ratio of $1.7$, and an asymptotic bound of $1.7$ was known for FF for many years \cite{johnson1974worst}). There are algorithms with smaller asymptotic competitive ratios, and the best possible asymptotic competitive ratio is known to be in $[1.5403,1.58889]$ \cite{van1992improved,Seiden02J,BBG}.
Other variants of bin packing where the sequence of items must remain ordered even for offline solutions include {\it Packing with LIB (largest item in the bottom)} constraints, where an item can be packed into a bin with sufficient space if it is no larger than any item packed into this bin \cite{M02,FM,manyem2003approximation,epstein2009online,DTY}.

In our algorithms, we say that a bin $B$ has color $c$ if the last item that was packed into $B$ has this color. Obviously, a bin changes its color as items are packed into it. For simplicity, we use names of colors as the elements of $\C$.
Another algorithm for black and white bin packing presented in \cite{waoa12} is the algorithm Pseudo. This algorithm keeps a list of pseudo-bins, each being a list of (valid) bins. Each new item is assigned to a pseudo-bin and then to a bin of this pseudo-bin.
The color of a (non-empty) pseudo-bin is defined to be the color of its last bin.
An item is first assigned to a pseudo-bin of the opposite color (that is, a white item to a black pseudo-bin and a black item to a white pseudo-bin), opening a new pseudo-bin for the item if this assignment is impossible (there is no pseudo-bin of the other color). A pseudo-bin is split into bins in an online fashion; a new item is packed into the last bin of the pseudo-bin where it was assigned (note that this is always possible with respect to the color of the item), and a new bin (for this pseudo-bin) is opened if the empty space in the current last bin of the pseudo-bin is insufficient. In the case that there are multiple pseudo-bins that are suitable for the new item (multiple pseudo-bins have the opposite color), then in principle any one of them is chosen (that is, the analysis holds for arbitrary tie-breaking), but the algorithm was defined such that such a bin of minimum index is selected.
A simple generalization of Pseudo for colorful packing is to assign a new item to a pseudo-bin of a minimum index whose color is different from the color of the new item. We show that this algorithm has an unbounded (absolute and asymptotic) competitive ratio. We show, however, that the tie-breaking rule can be modified, and a variant of this algorithm, called {\sc Balanced-Pseudo} ($BaP$), has an absolute (and asymptotic) competitive ratio of $4$. Roughly speaking, $BaP$ tries to balance the colors of pseudo-bins; for a new item it finds the most frequent color of pseudo-bins (excluding the pseudo-bins having the same color as the new item), and assigns the new item to such a pseudo-bin. Interestingly, this approach is much more successful.

Finally, we design two new lower bounds. We give a lower bound of $2$ on the asymptotic (and absolute) competitive ratio of any algorithm. This last lower bound is valid already for $|\C|=2$ (i.e., for black and white bin packing) and it significantly improves the previous lower bound of approximately 1.7213 \cite{waoa12}. We also consider zero size items. It was shown in \cite{waoa12} that Pseudo is an optimal algorithm for zero size items (its absolute competitive ratio is $1$). We show that in contrast, if $|\C|\geq 3$, then the asymptotic competitive ratio of any algorithm for such items is at least $\frac 32$. This implies that the two problems (colorful bin packing and black and white bin packing) are different.

In Section \ref{alg} we demonstrate that the existing algorithms have poor performance, we define algorithm $BaP$, analyze its competitive ratio for arbitrary items and for zero size items, and show that the analysis is tight. Lower bounds for arbitrary online algorithms are given in Section \ref{lbs}. 

\section{Algorithms}
\label{alg}
We start this section with examples showing that the algorithms that had a good performance for black and white bin packing (or their natural generalizations, all defined in the introduction) have a poor performance for colorful packing.

\begin{proposition}\label{FFWFBF}
The algorithms FF, BF, WF, AF, and Pseudo have unbounded asymptotic competitive ratios for colorful bin packing.
\end{proposition}
\begin{proof}
Let $M \geq 4$ be a large integer, and consider the following input. The input is presented in phases. Each phase consists of $M$ white items, followed by $2M$ items, each of which is either red or blue, with alternating colors, i.e., the colors alternate between red and blue, starting with a red item, and there are $M$ items of each color in each phase.  There are $N$ phases in total for a large integer $N \geq 2$.  Let $\eps=\frac{1}{N^2M^2}$, $\del=\eps^2=\frac{1}{N^4M^4}$. We will define the sizes of items differently for the different algorithms. Item sizes will be in $(0,\eps]$. The total size of all items will not exceed $3MN\eps = \frac{3}{MN}<1$, thus the valid solutions are those where no bin contains two items of one color that are packed consecutively. An optimal solution uses $M$ bins, and it packs one item of each color into each bin in each phase (in each phase, a white item, a red item, and a blue item are packed into each bin of OPT in this order).

Since FF, BF, and WF are  specials case of AF, the property for AF will follow from the examples given for those algorithms.
Consider FF, BF, and Pseudo (defined such that in a case of a tie it chooses the pseudo-bin of the minimum index), all items have sizes of $\eps$. FF acts as follows. In the first phase, each white item is packed into a different bin, and all red and blue items are packed into the first bin. In any further phase, the first white item is packed into the first bin, any additional white item requires a new bin, and the red items and blue items are packed into the first bin. The total number of bins is $M+(N-1)(M-1)=NM-N+1$.
Pseudo will act in the same way as FF, only it assigns the items to pseudo-bins, and each non-empty pseudo-bin only has one bin. BF packs the white items of the first phase into $M$ bins, then it packs a red item into one of its bins, and all further items that are not white will be packed into this bin as well. In each phase, one white item is added to the bin that contains red and blue items (in addition to a few white items), and this bin will always have the largest total size of items. The remaining white items are always packed into new bins, and therefore the resulting number of bins is $NM-N+1$ again. Finally, for WF, the sizes of the last white $M-1$ items of each phase are $\eps$, and the size of any other item is $\del$. Note that the total size of the $N$ first white items of the $N$ phases plus all red items and all blue items is $(2M+1) N\del \leq 3MN\del =\frac{3}{M^3N^3} < \frac {1}{M^2N^3}<\eps$. Thus, whenever it is possible, items will be packed into a bin that does not contain an item of size $\eps$. All items except for the larger white items (whose sizes are $\eps$) will be packed into the first bin, resulting in $NM-N+1$ bins.
For sufficiently large $M$, the competitive ratio is at least $N$, and thus the asymptotic competitive ratios of these algorithm are unbounded.
\end{proof}

\subsection{A new algorithm}
We define an algorithm called {\sc Balanced-Pseudo} ($BaP$). The algorithm keeps a sequence of pseudo-bins denoted by $P_1,P_2,\ldots$, where each pseudo-bin is a sequence of bins. For pseudo-bin $P_j$, its sequence of bins is denoted by $B^j_1,$$B^j_2,$$\ldots,$$B^j_{n_j}$. Let $k$ denote the number of pseudo-bins (at a given time). For any $1 \leq j \leq k$, $C_j$ denotes the color of the last item assigned to $P_j$ (this will be the color of the last item of $B^j_{n_j}$), and it is called the color of the pseudo-bin $P_j$.

Algorithm $BaP$ is similar to algorithm Pseudo \cite{waoa12}, but
it tries to balance the number of pseudo-bins of different colors,
and it prefers to assign an item to a pseudo-bin of a color that
occurs a maximum number of times (excluding  pseudo-bins having
the same color as the new item). For a new item $i$, if all
pseudo-bins have the color $c_i$, then a new pseudo-bin $P_{k+1}$
is opened, where it consists of one bin  $B^{k+1}_1$. In this
case, we let $k=k+1$, $n_k=1$. Otherwise, for any color $g \neq
c_i$, let $N_g$ be the number of pseudo-bins of color $g$. Let
$g'$ be a color for which $N_{g'}$ is maximal. Assign item $i$ to
a pseudo-bin $P_j$ of color $g'$. If $i$ can be packed into
$B^j_{n_j}$ (with respect to the total size of items, as by
definition the color of $P_j$ is $g' \neq c_i$, so the color of
$i$ does not prevent its packing), then add it to this bin (as its
last item), and otherwise, let $n_j=n_j+1$, and pack $i$ into
$B^j_{n_j}$ as its only item. For all cases, if $i$ was assigned
to pseudo-bin $P_j$, then let $C_j=c_i$ (this is done no matter
how $j$ is chosen).

\subsection{Analysis} The analysis separates the effect of
sizes from the effect of colors. This is possible since $BaP$
(similarly to Pseudo) already has such a separation. The number of
pseudo-bins is independent of the sizes of items, while the
partition of a pseudo-bin into bins is independent of the colors.
The algorithm that is applied on every pseudo-bin is simply NF,
and moreover, a new bin is used when there is no space for the
current item in the previous bin of the same pseudo-bin. Every
pair of consecutive bins of one pseudo-bin have items whose total
size exceeds $1$, thus the resulting bins are occupied by a total
size above $\frac 12$ on average, possibly except for one bin of
each pseudo-bin. We show that at each time that a new pseudo-bin
is opened, an optimal solution cannot have less than half the
number of bins, even if items have zero sizes. Informally, the
reason is that a new pseudo-bin is opened when all pseudo-bins
have the color of the new item. However, once the number of
pseudo-bins of this color exceeds half the number of pseudo-bins,
$BaP$ prefers to use such bins as much as possible (in this case
their number decreases), and an increase in their number can only
be caused by an input where there is a large number of items of
the same color arriving almost consecutively. Obviously, such
inputs require large numbers of bins in any solution.

We let $LB_0=\sum_{i=1}^n s_i$. Obviously, $OPT \geq LB_0$. Let $1 \leq i \leq j \leq n$. For
any color $c$ that appears in the subsequence of consecutive $j-i+1$ items
$i,i+1,...,j$, let $C(i,j,c)$ be the number of times that it
appears. Let
\begin{equation}\label{eqa} LB(i,j,c)=C(i,j,c)-(j-i+1-C(i,j,c))=2C(i,j,c)-j+i-1 \ , \end{equation}
$LB(i,j)=\max_c LB(i,j,c)$, and $LB_1=\max_{i,j} LB(i,j)$. For any
non-empty input we have $LB_1 \geq 1$ since $LB(i,i,c_i)=1$ for
any $i$. Note that $LB(i,j,c)$ is positive only if the number of
times that $c$ appears in the subsequence $i,\ldots,j$ is more
than $\frac{j-i+1}2$ (i.e., more than half the items of this
subsequence are of color $c$), and thus for computing $LB_1$ it is
sufficient to consider for every subsequence only a color $c$ that
appears a maximum number of times in this subsequence. The
following lemma generalizes a property proved in \cite{waoa12}.

\begin{lemma}\label{LBone}
$OPT \geq LB_1$.
\end{lemma}
\begin{proof}
We prove $OPT \geq LB(i,j,c)$ for any $1 \leq i\leq j \leq n$. If $c$ appears at most $\frac{j-i+1}2$ times, then we are done as $LB(i,j,c)\leq 0$.
Otherwise, consider an optimal solution for the input. Remove the items $1,2,\ldots,i-1$ one by one in this order from OPT. Each removed item must be the first item in its bin in the packing resulting from removing its preceding items. Thus, the packing remains valid. Similarly, remove the items $n,n-1,\ldots,j+1$ one by one. Each removed item must be the last item of its bin in the packing that results from removing its succeeding items, and the packing remains valid. Some of the bins may become empty. Let $L \leq OPT$ be the remaining number of bins. Recall that these bins contain $j-i+1$ items in total. Since there is an item of a color that is not $c$ between every two items of color $c$, a bin that contains $x$ items can contain at most $\frac{x-1}2+1=\frac x2+\frac 12$ items of color $c$, and thus $\frac{j-i+1+LB(i,j,c)}2=C(i,j,c)\leq \frac{j-i+1}2+\frac L2$, implying that $OPT \geq L \geq LB(i,j,c)$.
\end{proof}

Consider the action of $BaP$, and let $k$ be the index of the last pseudo-bin (i.e., $k$ is the final value of the variable $k$).
For $1 \leq m \leq k$, let $LB^m$ denote $LB_1$ at the time that the first item is assigned to $P_m$. Let $Y_m$ be the (index of the) first item that is assigned to $P_m$, and let $X_m$ be its color (thus $Y_1=1$ holds by definition, i.e., the first item of the input is also the first item assigned to the first pseudo-bin). For convenience, let $Y_{k+1}=n+1$. Let phase $m$ be the subsequence of consecutive items $Y_m,\ldots,Y_{m+1}-1$. In the lemmas below, when we discuss properties holding during phase $m$, we mean that they hold starting the time just after
$Y_m$ is packed and ending right after $Y_{m+1}-1$ is packed.

\begin{theorem} For any $1 \leq m \leq k$, there exists $i\leq Y_m$ such that $C(i,Y_m,X_m) \geq \frac {m+3}4+\frac{Y_m-i}2$.
\end{theorem}
\begin{proof}
We prove the claim by induction.  For $m=1$, $Y_m=1$, and
$C(1,1,c_1)=1$ as required. For $m=2$, the items $Y_2$ and $Y_2-1$
have the same color $X_2$ (as $Y_2-1$ was assigned to $P_1$ and
$Y_2$ is assigned to $P_2$). Thus, we find $C(Y_2-1,Y_2,X_2)=2$. Next, assume that the claim
holds for some $m \geq 2$. We will prove the claim for $m+1$ by
considering phase $2 \leq m \leq k-1$.

\begin{lemma}\label{first}
If at some time in phase $m$ (where $2 \leq m \leq k-1$) an item
$i$ of a color that is not $X_{m+1}$ is assigned to a pseudo-bin
of a color that is not $X_{m+1}$ (the two last items that the
pseudo-bin receives are of colors different from $X_{m+1}$), then
just before assigning $i$ (the second item out of the two items
whose colors are not $X_{m+1}$) there are less than $(m+1)/2$
(that is, at most $m/2$) pseudo-bins of color $X_{m+1}$.
\end{lemma}
\begin{proof}
During phase $m$, there are $m$ pseudo-bins. Assume by contradiction that when $i$ arrives there are at least $(m+1)/2$ pseudo-bins of color $X_{m+1}$. Since the color of $i$ is different, by the definition of $BaP$, $i$ must be assigned to a pseudo-bin of this color, as the number of pseudo-bins of any other color cannot exceed $(m-1)/2<(m+1)/2$. A contradiction.
\end{proof}

\begin{lemma}\label{samecol}
\label{seco} If during phase $m$
there are always at least $(m+1)/2$ pseudo-bins of color
$X_{m+1}$, then $X_m=X_{m+1}$. In this case, letting $t$ be the
number of items of color $X_m$ in phase $m$, phase $m$ contains
$t-1$ items of other colors.
\end{lemma}
\begin{proof}
When phase $m$ starts, just after $Y_m$ is packed, there are $m$ pseudo-bins of color $X_m$ and no pseudo-bins of any other color. Since there are at least $(m+1)/2 \geq 1$ pseudo-bins of color $X_{m+1}$ at this time, we find $X_m=X_{m+1}$. During the phase, any item whose color is not $X_{m+1}$ is assigned to a pseudo-bin whose color is $X_{m+1}$. After item $Y_{m+1}-1$ is packed, once again there are $m$ pseudo-bins of color $X_m$ (since when $Y_{m+1}$ is packed, the number of pseudo-bins of color $X_m$ becomes $m+1$).
Thus, during the phase, starting the time right after $Y_m$ is packed, every pseudo-bin receives the same number of items of color $X_{m+1}=X_m$ and other colors. Since $c_{Y_m}=X_m$, the number of items of color $X_m$ is larger by $1$ than the number of items of other colors out of the items of phase $m$.
\end{proof}

If the condition of Lemma \ref{seco} holds, then let $i$ be such
that $C(i,Y_m,X_m) \geq \frac {m+3}4+\frac{Y_m-i}2$, and let $t$ be
the number of items of color $X_m=X_{m+1}$ in phase $m$. We have
$C(i,Y_{m+1},X_{m+1}) \geq \frac {m+3}4+\frac{Y_m-i}2+t$, and
$Y_{m+1}-Y_m=2t-1$. Thus, $C(i,Y_{m+1},X_{m+1}) \geq \frac
{m+3}4+\frac{Y_m-i}2+\frac{Y_{m+1}-Y_m+1}2 > \frac
{(m+1)+3}4+\frac{Y_{m+1}-i}2$ as required.

\begin{lemma}
If there is a time in phase $m$ that
at most $m/2$ bins were of color $X_{m+1}$, then there exists an
index $i$ such that $Y_m \leq i \leq Y_{m+1}-1$ where
$C(i,Y_{m+1},X_{m+1}) \geq \frac {m+4}4+\frac{Y_{m+1}-i}2$.
\end{lemma}
\begin{proof}
Consider the last time during phase $m$ that there are at most
$m/2$ bins of color $X_{m+1}$, and let $i$ be the first item right
after this time. Since after item $Y_{m+1}-1$ arrives, all $m$
pseudo-bins have color $X_{m+1}$ and $m>m/2$, the time just after
$Y_{m+1}-1$ arrives does not satisfy the condition, so the last
such time must be earlier, $i$ is well-defined, and $i \leq
Y_{m+1}-1$. We have $c_i=X_{m+1}$ as its assignment to a
pseudo-bin increased the number of pseudo-bins of this color.
Moreover, starting this time, there are at least $(m+1)/2$ bins of
color $X_{m+1}$ at all times until after the arrival of $Y_{m+1}$
(by the choice of the time, and since $Y_{m+1}$ has the same color
and causes the creation of a new pseudo-bin of this color). If $m$
is even, then just before $i$ is packed, there are exactly $m/2$
pseudo-bins of color $X_{m+1}$ and $m/2$ pseudo-bins of other
colors, and after item $Y_{m+1}$ is assigned, there are $m+1$
pseudo-bins of color $X_{m+1}$. Moreover, while the items
$i,\ldots,Y_{m+1}-1$ are being assigned, every item whose color is
not $X_{m+1}$ is assigned to a pseudo-bin of color $X_{m+1}$, so
every pseudo-bin receives alternating colors (items of color
$X_{m+1}$ alternate with other colors). Thus, if there are $t$
items whose colors are not $X_{m+1}$ among these items, there are
$t+\frac m2$ items of color $X_{m+1}$, and the total number of
items is $Y_{m+1}-i=2t+\frac m2$. Including $Y_{m+1}$, we have
$C(i,Y_{m+1},X_{m+1})=t+\frac m2+1=\frac m2+1
+\frac{Y_{m+1}-i}2-\frac m4=\frac {(m+1)+3}4+\frac{Y_{m+1}-i}2$ as
required. If $m$ is odd, then if there are $t$ items whose colors
are not $X_{m+1}$ among these items, there are $t+\frac {m+1}2$
items of color $X_{m+1}$, and the total number of items is
$Y_{m+1}-i=2t+\frac {m+1}2$. We have $C(i,Y_{m+1},X_{m+1})=t+\frac
{m+1}2+1=\frac m2+\frac 32 +\frac{Y_{m+1}-i}2-\frac {m+1}4 > \frac
{m+4}4+\frac{Y_{m+1}-i}2$ as required. \end{proof} This
completes the proof of the theorem. \end{proof}

The next corollary follows from choosing
 $j=Y_k$ and $i$ such that $C(i,Y_k,X_k)$ $\geq \frac{m+3}4+\frac{Y_m-i}2$, and using (\ref{eqa}).
\begin{corollary}
We have $LB_1 \geq LB^k  \geq LB(i,Y_k,X_k) \geq \frac {k+1}2$.
\end{corollary}

\begin{corollary}\label{4n2}
The absolute competitive ratio of $BaP$ is at most $4$ for arbitrary items, and at most $2$ for zero size items. 
\end{corollary}
\begin{proof}
For zero size items, $BaP$ produces exactly $k \geq 1$ bins. We find that $k \leq 2LB_1-1 <2\cdot OPT$.
Consider an input consisting of arbitrary items. Every two consecutive bins resulting from one pseudo-bin of $BaP$ have a total size of items that exceeds $1$. Thus, for a pseudo-bin that results in $x$ bins, the total size of items is above $\lfloor   \frac x 2 \rfloor \geq \frac{x-1}2$. We find that the total size of items is at least $\frac{BaP}2-\frac k2$. Thus, $LB_0 \geq \frac{BaP}2-\frac k2$ while $LB_1 \geq \frac{k+1}2$.
We find that $BaP \leq 2LB_0+k \leq 2LB_0+2LB_1-1 < 4\cdot OPT$.
\end{proof}

We can show that the analysis of $BaP$ is tight.
\begin{proposition}
The asymptotic competitive ratio of $BaP$ is at least $2$ for zero size items, and at least $4$ for arbitrary items.
\end{proposition}
\begin{proof}
We will use the following parameters. Let $N\geq 2$ be a large integer. Let $M=4^{N+1}$, let $a_1=1$, and for $i>1$, let $a_i=(3a_{i-1}+2)/4$.
\begin{lemma}\label{simprop} We have $1\leq a_i<2$, $a_i>a_{i-1}$ for all $i$, and $\lim_{i \rightarrow \infty} a_i =2$. Moreover, $a_i=2-(3/4)^{i-1}$ holds.
\end{lemma}
\begin{proof}
We prove the first part by induction. It holds for $i=1$. Assume that it holds for $i-1$ (for some $i>1$). We have that$(3a_{i-1}+2)/4 \geq 1$ holds since $a_{i-1} \geq 1$, and  $(3a_{i-1}+2)/4 < 2$ holds since $a_{i-1}<2$.
For the second part, we find $4a_i=3a_{i-1}+2$, or equivalently $2>a_i=2-3(a_i-a_{i-1})$, that is, $a_i>a_{i-1}$.
Let $b_i=2-a_i$. We have $4(2-b_i)=3(2-b_{i-1})+2$, or equivalently, $b_i=3b_{i-1}/4$, and $b_i=(3/4)^{i-1}$ since $b_1=1$. Therefore $a_i=2-(3/4)^{i-1}$, and since the sequence $b_j$ tends to zero as $j$ tends to infinity, $a_j$ tends to $2$.
\end{proof}

We start with an input of zero size items. In this input all items
are white, red, or blue. The input consists of the following $N+1$
phases. In phase $0$, $M$ white items arrive. In phase $i$ (for $1
\leq i \leq N$), $a_i\cdot M/2$ red items arrive, and then
$(1-a_i/2)M$ blue items arrive. We find $a_i\cdot
M/2=(2-(3/4)^{i-1})4^{N+1}/2=2(4^{N}-3^{i-1}\cdot4^{N-i+1})$, and
$(1-a_i/2)M=2\cdot 4^N-2\cdot4^{N}+2\cdot3^{i-1}\cdot4^{N-i+1}$.
The numbers of red and blue items are even integers in $(0,M)$,
and their sum is $M$. Phase $i$ ends with the arrival of $M$ white
items. We have $OPT=M$. Obviously, $M$ bins are needed already for
the first $M$ white items. Each bin of the optimal solution
receives one white item in phase $0$, and in each additional phase
it receives one red item or one blue item, and additionally one
white item.

\begin{lemma} After $i$ phases $BaP$ has $a_{i+1}M$ pseudo-bins, all of which are white. \end{lemma}
\begin{proof} By induction. This holds for $i=0$.
Assume that it holds after phase $i-1$. In phase $i$, first the red items are assigned to distinct pseudo-bins, and now there are $a_i\cdot M/2$ red pseudo-bins and $a_i\cdot M/2$ white pseudo-bins. Now the blue items are packed such that half of them join red pseudo-bins and half join white pseudo-bins. The number of white pseudo-bins is now $a_i\cdot M/2-(1-a_i/2)M/2=M/4(3a_i-2)$. The number of pseudo-bins that are either red or blue is now $a_i\cdot M/2+(1-a_i/2)M/2=M(a_i+2)/4$. Note that $(a_i+2)/4<1$ since $a_i<2$. The $M$ white items can join $M/4(a_i+2)$ pseudo-bins that are either red or blue, and the remaining $M-M/4(a_i+2)$ items cause the opening of new white pseudo-bins. The total number of pseudo-bins now is $a_i\cdot M+(M-M(a_i+2)/4)$ and they are all white. The last number is equal to $M(a_i+1-a_i/4-1/2)=M(3a_i+2)/4=M\cdot a_{i+1}$.\end{proof}

We find that  after $N+1$ phases, the algorithm has $(2-(3/4)^{N}) \cdot M$ pseudo-bins, each consisting of one bin, which implies the lower bound.

In order to prove that the asymptotic competitive ratio is at least $4$ for arbitrary item sizes, we start with presenting the input above to $BaP$. At this time, all items are of three colors and have zero sizes, $OPT=M$, the algorithm has $2M-m$ pseudo-bins where $m=(\frac 34)^N M$. The input continues as follows (we ensure that $OPT=M$ will hold for the complete input). There are $2M-m-1$ items, all of different new colors (none of these colors is white or red or blue). Moreover, we reserve the color black for later, and thus we require that none of these colors is black. Each of these items has size $2\eps$ (for some $\eps<1/(8M)$). $OPT$ will use one bin for items of size $2\eps$, while $BaP$ will assign each item to a different pseudo-bin. Now all the bins of $BaP$ have different colors (one pseudo-bin remains white). Next, $M-1$ black items arrive, where each item has size $1-\eps$. $OPT$ adds them to its white bins, the algorithm assigns at most one item to a white pseudo-bin, so at least $M-2$ items are assigned to different pseudo-bins whose color was not white, red, blue, or black (and the last item assigned to this pseudo-bin had size $2\eps$). Thus, there are at least $M-1$ black pseudo-bins, and at least $M-2$ of them consist of two bins each, as the total size of items assigned to it is above $1$. Next, there are $M-2$ items all of different and new colors and sizes of $2\eps$. $OPT$ packs them into the bin that already has items of this size, while the algorithm adds them to its black pseudo-bins, and at least $M-3$ pseudo-bins now consist of three bins.
The algorithm will have at least $2M-m+(M-2)+(M-3)=4M-m$ bins, while $OPT=M$. The competitive ratio approaches $4$ for a sufficiently large value of $N$.

Note that this example does not require any assumptions regarding the behavior of $BaP$ in cases of ties.
The example requires, however, a large number of different colors. We provide a different example that is valid for a run of $BaP$ where ties between pseudo-bins of one color are broken in favor of smaller indices, and $\C=\{white,red,blue\}$.
Once again, the input starts with the items of zero size as above.  Afterwards, there are three batches of items, consisting of $M$ blue items, $M$ white items, and $M$ blue items, respectively, of sizes that we will define. Since the number of pseudo-bins is above $M$ and all of them are white, blue items must join white pseudo-bins, and white items must join blue pseudo-bins.  The three batches are packed into the first $M$ pseudo-bins, where the $j$th item of a batch is packed into the pseudo-bin of index $j$.
For $1 \leq t \leq M+1$, let
$\del_{t}=\eps/4^{t}$ (thus we have $\del_{t+1}=\del_t/4)$.
The size of the $t$th item in the first batch (of blue items) is $\del_t$ ($t=1,...,M$).
The size of the $t$th item in the second batch (of white items) is $1-3\del_{t+1}$ ($t=1,...,M$).
The size of the $t$th item in the third batch (of blue items) is $\del_t$ ($t=1,...,M$).
We have $\del_t+(1-3\cdot\del_{i+1})>1$ since
$\del_t-3\cdot\del_{t+1}=\del_t/4$. Therefore, each pseudo-bin $t=1,\ldots,M$ consists of three bins.

We show that for this input $OPT \leq M+2$. Given the packing into $M$ white bins, for $t=1,...,M-1$
we group the items of sizes $\del_t, 1-3\cdot\del_t,
\del_t$ (of colors blue, white, and blue, respectively) and pack them into $M-1$ bins. A blue item of size $\del_{M}$ is added to the remaining bin, and the two items of sizes $\del_M$ and $1-3\cdot\del_{M+1}$ are packed into new bins.
\end{proof}

\section{Lower bounds}
\label{lbs}

The (absolute or asymptotic) competitive ratio cannot decrease if the cardinality of $\C$ grows. Thus, when we claim a negative result for $|\C| \geq \ell$, it is sufficient to prove it for $|\C|=\ell$. Thus, the lower bound for arbitrary items is proved for $|\C|=2$, and the lower bound for zero size items is proved for $|\C|=3$.

\subsection{An asymptotic lower bound of 2}
We will consider an algorithm, and construct an input consisting of black and white items based on its behavior. The construction is carried out in phases, where in each phase the algorithm has to pack a black item after a white item. If they are packed together, it turns out that it would have been better to pack this last black item separately, since another smaller black item arrives, and a large white item that should have been combined with the first black item of this phase. Since no other combination is possible, the algorithm has two new bins instead of just one. If the algorithm uses a new bin for the first black item, it turns out that the phase ends, and the algorithm used a new bin when this was not necessary. The first situation is slightly better for the algorithm, and a ratio of $2$ will follow from that. The precise construction is presented in the proof of the following theorem.

\begin{theorem}
The asymptotic competitive ratio of any algorithm for colorful bin packing is at least $2$.
\end{theorem}
\begin{proof}
Consider an online algorithm $A$. Let $N>3$ be a large integer. Let $\eps=\frac 1{N^3}$, and
$\del_i=\frac 1{5^i \cdot N^3}$ for $1 \leq i \leq N^2$. Let $\C= \{\mbox{ black, white}\}$.
The list of items will consist of white items called {\it regular white items}, each of size $\eps$,
white items called {\it huge white items}, whose sizes are either of the form $1-2\del_i$ (for
some $1 \leq i \leq N^2$) or $1$, black items called {\it special black items},
whose sizes are of the form $3\del_i$, and black items called {\it regular black items}
whose sizes are of the form $\del_i$.

The list is created as follows. An index $i$ is used for the
number of regular white items that have arrived so far (each such item is
followed by a regular black item). An index $j$ is used for the
number of huge white items that have arrived so far (each such item is
preceded by a black item and followed by a black item). The input stops when
one of $i=N^2$ and $j=N$ happens (even if the second event did not happen). Let $i=0$ and $j=0$. 

\noindent {\bf 1.} If $j=N$, then stop. Else, if $i=N^2$, then $N-j$ huge white items of size $1$ each
arrive; stop.

\noindent {\bf 2.} Let $i=i+1$; a regular white item arrives; a regular
black item of size $\del_i$ arrives.

\noindent {\bf 3.} If the last black item is packed into a new bin, the phase ends. Go to step 1 to start a new phase. 

\noindent {\bf 4.} Else, it must be the case that the last black item is packed into a bin where the
last item is white. Let $j=j+1$, a special black item of size
$3\del_i$ arrives, then a huge white item of size $1-2\del_i$ arrives,
and finally, a regular black item of size $\del_i$ arrives, and the phase ends. Go to step 1 to start a new phase.

%
%
%
%
%

\begin{lemma}\label{easyp}
Any huge white item is strictly larger than $1-\eps$. Any black
item is strictly smaller than $\eps$. The total size of a huge white item of phase $i$ and a black item of an earlier phase is above $1$.
\end{lemma}
\begin{proof}
The largest black item can be of size $3\del_1 <\eps$. The
smallest huge white item can have the size $1-2\del_1>1-\eps$. Finally consider a huge white item of phase $i_2$ and a black item of phase $i_1<i_2$. The size of the white item is $1-2\del_{i_2}$, and the size of the black item is at least $\del_{i_1}$. We have
$(1-2\del_{i_2})+\del_{i_1}=1+\frac{1}{N^3}(\frac{1}{5^{i_1}}-\frac{2}{5^{i_2}}) \geq 1+\frac{1}{N^3}(\frac{1}{5^{i_1}}-\frac{2}{5^{i_1+1}})>1$.
\end{proof}

\begin{lemma}
$N \leq OPT \leq N+1$.
\end{lemma}
\begin{proof}
There are $N$ huge white items, each of size above $\frac 12$,
thus, since a pair of such items cannot be packed into a bin together even with a black item, $OPT \geq N$.
We create a packing with $N+1$ bins as follows. If there are huge
white items of size $1$, each such item is packed into a separate
bin. We show how the remaining items can be packed into $j$ bins (where $j$ is the final value of the variable $j$).
Every remaining huge white item is packed in a bin with the last
regular black item that arrived before it, and the regular black
item that arrived after it. The total size of such three items of phase $i$ is $1$.
This leaves a sequence of items of
alternating colors, where some of the black items are special. The white items in the remaining input are regular, and the black item of phase $i$ has a size of either $\del_i$ or $3\del_i$. In
this sequence, every item is no larger than $\eps$, and there are
$2i \leq 2N^2$ items (where $i$ is the final value of this variable). Thus, the total size of these items is below
$1$, and they are all packed into a single bin.
\end{proof}

\begin{lemma}\label{howmany}
The number of bins used by the algorithm up to a time when $i=i'$
is at least $i'$. The number of black bins at a time when $j=j'$
is at least $2j'+1$.
\end{lemma}
\begin{proof}
In a step where $i$ increases but $j$ does not increase, the black
regular item is packed into a new bin. In a step where both $i$
and $j$ increase, the huge white item must be packed into a new
bin as the only black item that arrived so far and fits into a bin
with this last white item is the last regular black item (since black items of earlier phases are too large, and the last special black item has size $3\del_i$), but this
item was packed into a bin that already has a white item, so its total size of items is above $\eps$, and the huge white item cannot be packed
there. This proves the first claim, since in both cases at least one new bin is used.

The second claim is proved via induction. First note that when a pair of a regular white
item and a regular black item arrive, the number of black bins cannot
decrease (no matter if they are packed into the same bin or not).
Moreover, when a huge white item of size $1$ arrives, it cannot be
packed into a non-empty bin as all item sizes are strictly
positive, so it cannot change the number of black bins either.
Consider the packing as long as $j=0$. After step 3 was applied once or more, there is at least one black bin that contains the last black regular item. Each time that
$j$ increases, since the huge white item is packed into a previously empty
bin, the two items arriving just before and just after the white
item (the special item and the regular item) increase the
number of black bins by $2$, since the special item is either
packed into a new bin or into a white bin, the huge white bin does
not change the status of a previously non-empty bin, so its
packing does not change the number of black bins, and the regular
item also increases the number of black bins by $1$.\end{proof}

For a fixed value of $N$, if the input was terminated since
$i=N^2$ but $j<N$, then the cost of the algorithm is at least
$N^2+N-j \geq N^2+1$. As $OPT \leq N+1$, we find a competitive
above $N-1 > 2$. If $j=N$, then the cost of the algorithm is at
least $2N+1$ (as this is a lower bound on the number of black
bins), while $OPT \leq N+1$, and we find a ratio of at least
$2-\frac{1}{N+1}$. We found that for any $N>3$, there is an input
where $OPT\geq N$, and the competitive ratio for this input is at
least $2-\frac{1}{N+1}$. This implies the claim.
\end{proof}

\subsection{A lower bound for zero size items}
It was shown in \cite{waoa12} that if all items have zero sizes, then the algorithm Pseudo finds an optimal solution (that is, its absolute competitive ratio is $1$). Our analysis of $BaP$ implies that its absolute and asymptotic competitive ratios for zero size items are equal to $2$. Here, we show that there cannot be an online algorithm for colorful bin packing with at least three colors and zero size items that produces an optimal solution (a solution that uses the minimum number of bins).

\begin{theorem}
Any algorithm for zero size items with $|\C|\geq 3$ has an asymptotic competitive ratio of at least $\frac 32$.
\end{theorem}
\begin{proof}
We will use $\C=\{\mbox{white, red, blue}\}$. Recall that all items have zero sizes, thus for every presented item we only specify its color. Let $M \geq 2$ be a large integer. We construct an input for which $M \leq OPT \leq M+3$. The input starts with phase $0$ that consists of $M$ white items. Thus, $OPT \geq M$. The remainder of the input is presented in phases. In parallel to presenting the input, we will create a packing $\pi$ for the complete input. This  packing will consist of $M+3$ bins. The $M$ items of phase $0$ are packed in $\pi$ into $M$ bins called regular bins. In addition to the $M$ regular bins of $\pi$, there will be a special bin of each color in $\pi$ (this bin is empty after phase $0$). The regular bins of $\pi$ ($M$ bins in total), will always be of one color (this color can be any of the three colors).
Each phase $i$ will have a color $G(i)$ associated with it. This is the color of the $M$ regular bins of $\pi$. The color associated with phase $0$ is white.

Phase $i$ is defined as follows. Let $c_i$ and $c'_i$ be the two colors that are not the color associated with phase $i-1$ (i.e., $c_i,c_{i'}\in \C\setminus \{G(i-1)\}$, $c_i\neq c_{i'}$.
There are $2M$ items of alternating colors; the items of odd indices are of color $c_i$, and the items of even indices are of color $c'_i$.
Let $W_{i}$, $R_{i}$, and $B_{i}$, be the numbers of white, red, and blue bins, that the algorithm has after the last $2M$ items have arrived. Phase $i$ ends with $M$ items of the color for which the number of bins of the algorithm is maximal after the $2M$ first items of phase $i$ have been packed by the algorithm (that is, letting $X=\max\{W_i,R_i,B_i\}$, the last $M$ items are white if $X=W_i$, otherwise, if $X=R_i$, then they are red, and otherwise they are blue). Let $G(i)$ be the color of the last $M$ items of phase $i$.

Let $N_i$ be the number of bins of the algorithm after phase $i$.
We have $N_0=M$. In phase $i \geq 1$ the algorithm obviously has at least $N_{i-1}$ bins after the first $2M$ items of phase $i$ have arrived, and there are at least $\frac{N_{i-1}}3$ bins of color $G(i)$. Therefore, after $M$ items of color $G(i)$ arrive, the algorithm has $M$ additional bins of color $G(i)$, and there are at least $\frac{N_{i-1}}{3}+M$ bins of color $G(i)$. We get $N_{i} \geq \frac{N_{i-1}}{3}+M$. Thus, $N_i \geq M \cdot \frac{3^{i+1}-1}{2\cdot 3^i}$. This holds for $i=0$ as $N_0=M$, and $\frac{3^1-1}{2\cdot 3^0}=1$, and using the recurrence, $N_{i+1} \geq (\frac{3^{i+1}-1}{2\cdot 3^{i}})M/3+M=(\frac{3^{i+2}-1}{2\cdot 3^{i+1}})M$.

Due to symmetry, we describe the packing $\pi$ for the case that the color associated with phase $i-1$ is white, and the first $2M$ items of phase $i$ alternate between red and blue (starting with red). If the last $M$ items of phase $i$ are blue or red, then the first $2M$ items are packed into the blue special bin (which remains blue), and the last $M$ items are packed into the $M$ regular bins.  If the last $M$ items are white, each bin receives a red item and an blue item. Now all regular bins are blue, and the last $M$ white items can be packed into them. The color associated with phase $i$ is indeed $G(i)$.

We find that the competitive ratio of the algorithm is at least $\frac{M}{M+3} \cdot \frac{3^{i+1}-1}{2\cdot 3^i}$. Letting $M$ and $i$ grow without bound we find a lower bound of $\frac 32$ on the asymptotic competitive ratio.
\end{proof}

\bibliographystyle{abbrv}

\begin{thebibliography}{10}

\bibitem{waoa12}
J.~Balogh, J.~B{\'e}k{\'e}si, G.~D{\'o}sa, L.~Epstein, H.~Kellerer, and
  Z.~Tuza.
\newblock Online results for black and white bin packing.
\newblock {\em Theory of Computing Systems}.
\newblock To appear.

\bibitem{BBG}
J.~Balogh, J.~B{\'e}k{\'e}si, and G.~Galambos.
\newblock New lower bounds for certain classes of bin packing algorithms.
\newblock {\em Theoretical Computer Science}, 440-441:1--13, 2012.

\bibitem{DS12}
G.~D\'osa and J.~Sgall.
\newblock First fit bin packing: A tight analysis.
\newblock In {\em Proc. of the 30th International Symposium on Theoretical
  Aspects of Computer Science (STACS2013)}, pages 538--549, 2013.

\bibitem{DTY}
G.~D{\'o}sa, Z.~Tuza, and D.~Ye.
\newblock Bin packing with ``largest in bottom" constraint: tighter bounds and
  generalizations.
\newblock {\em Journal of Combinatorial Optimization}, 26(3):416--436, 2013.

\bibitem{epstein2009online}
L.~Epstein.
\newblock On online bin packing with {LIB} constraints.
\newblock {\em Naval Research Logistics}, 56(8):780--786, 2009.

\bibitem{FM}
L.~Finlay and P.~Manyem.
\newblock Online {LIB} problems: Heuristics for bin covering and lower bounds
  for bin packing.
\newblock {\em RAIRO Operetions Research}, 39(3):163--183, 2005.

\bibitem{John}
D.~S. Johnson.
\newblock {\em Near-optimal bin packing algorithms}.
\newblock PhD thesis, MIT, Cambridge, MA, 1973.

\bibitem{J74}
D.~S. Johnson.
\newblock Fast algorithms for bin packing.
\newblock {\em Journal of Computer and System Sciences}, 8(3):272--314, 1974.

\bibitem{johnson1974worst}
D.~S. Johnson, A.~Demers, J.~D. Ullman, M.~R. Garey, and R.~L. Graham.
\newblock Worst-case performance bounds for simple one-dimensional packing
  algorithms.
\newblock {\em SIAM Journal on Computing}, 3:256--278, 1974.

\bibitem{LeeLee}
C.~C. Lee and D.~T. Lee.
\newblock A simple online bin packing algorithm.
\newblock {\em Journal of the ACM}, 32(3):562--572, 1985.

\bibitem{M02}
P.~Manyem.
\newblock Bin packing and covering with longest items at the bottom: Online
  version.
\newblock {\em The ANZIAM Journal}, 43(E):E186--E232, 2002.

\bibitem{manyem2003approximation}
P.~Manyem, R.~L. Salt, and M.~S.Visser.
\newblock Approximation lower bounds in online {LIB} bin packing and covering.
\newblock {\em Journal of Automata, Languages and Combinatorics},
  8(4):663--674, 2003.

\bibitem{Seiden02J}
S.~S. Seiden.
\newblock {On the online bin packing problem}.
\newblock {\em Journal of the ACM}, 49(5):640--671, 2002.

\bibitem{U71}
J.~D. Ullman.
\newblock The performance of a memory allocation algorithm.
\newblock Technical Report 100, Princeton University, Princeton, NJ, 1971.

\bibitem{van1992improved}
A.~van Vliet.
\newblock An improved lower bound for on-line bin packing algorithms.
\newblock {\em Information Processing Letters}, 43(5):277--284, 1992.

\bibitem{Ves}
P.~Vesel\'y.
\newblock Competitiveness of fit algorithms for black and white packing.
\newblock Manuscript, presented in MATCOS-13, 2013.

\end{thebibliography}

\end{document}